\documentclass[10pt,pre,onecolumn,superscriptaddress,reprint]{revtex4-1}
\usepackage{xr}
\usepackage{graphicx}
\usepackage{float}
\usepackage{amsmath}
\usepackage{amssymb}
\usepackage{mathtools} 
\usepackage{color}
\usepackage{xcolor}
\usepackage[english]{babel}
\usepackage{silence}
\usepackage[T1]{fontenc}
\usepackage[utf8]{inputenc}

\usepackage{algpseudocode}
\usepackage{algorithm}
\usepackage{amsthm}

\usepackage[columnwise]{lineno}
\usepackage{hyperref}
\usepackage{physics}
\usepackage[caption=false]{subfig}

\newtheorem{theorem}{Theorem}
\newtheorem{corollary}{Corollary}[theorem]

\newtheorem{definition}{Definition}

\usepackage{soul}
\setstcolor{red}

\begin{document}

\title{Local stability of spheres via the convex hull and the radical Voronoi diagram}

\date{\today}
\author{Peter K. Morse}
\thanks{Corresponding author.}
\email{peter.k.morse@gmail.com}
\affiliation{Department of Chemistry, Princeton University, Princeton, NJ 08544}
\affiliation{Department of Physics, Princeton University, Princeton, NJ 08544}
\affiliation{Princeton Institute of Materials, Princeton University, Princeton, NJ 08544}
\author{Eric Corwin}
\thanks{Corresponding author.}
\email{ecorwin@uoregon.edu}
\affiliation{Department of Physics, University of Oregon, Eugene, OR 97403}
\affiliation{Materials Science Institute, University of Oregon, Eugene, OR 97403}

\begin{abstract}
Jamming is an emergent phenomenon wherein the local stability of individual particles percolates to form a globally rigid structure. However, the onset of rigidity does not imply that every particle becomes rigid, and indeed some remain locally unstable. These particles, if they become unmoored from their neighbors, are called \textit{rattlers}, and their identification is critical to understanding the rigid backbone of a packing, as these particles cannot bear stress. The accurate identification of rattlers, however, can be a time-consuming process, and the currently accepted method lacks a simple geometric interpretation. In this manuscript, we propose two simpler classifications of rattlers based on the convex hull of contacting neighbors and the maximum inscribed sphere of the radical Voronoi cell, each of which provides geometric insight into the source of their instability. Furthermore, the convex hull formulation can be generalized to explore stability in hyperstatic soft sphere packings, spring networks, non-spherical packings, and mean-field non-central-force potentials.
\end{abstract}

\maketitle

\section{Introduction}

A rigid structure is one which holds its shape when perturbed infinitesimally. If this structure consists of particles, this rigid structure is said to be jammed~\cite{stillinger_systematic_1964, liu_jamming_1998, torquato_random_2000, ohern_jamming_2003, torquato_jammed_2010, liu_jamming_2010, hecke_jamming_2009}. While the system as a whole may be rigid, local regions of it may still be unconstrained.  The particles---or clusters of particles---making up these locally unconstrained regions are generally termed ``rattlers''~\cite{stillinger_systematic_1964, speedy_random_1998, rattlerHistory} and are removed from the consideration of the structure for many analyses. 

The rigorous rattler detection scheme in the literature~\cite{donev_linear_2004} relies on linear programming and is both computationally expensive and lacks a simple geometric interpretation. Another, based on an event-driven packing protocol, gives direct physical meaning to rattler detection by using a stability analysis to systematically prune compressive forces, leaving rattlers fully unconstrained~\cite{lerner_simulations_2013}. However, this method scales poorly with system size and dimension, as it requires matrix inversion. These methods are, however, exact, and the resulting stable networks which they find are identical. In light of the complexity of these algorithms, a naive rattler detection scheme via constraint counting has proliferated and been used widely as a proxy, despite its shortcomings. The naive algorithm exploits the fact that the minimum number of constraints necessary to stabilize a particle in $d$ dimensions is $d+1$.  Thus, the number of contacts on each particle is counted, and those with fewer than $d+1$ contacts are deemed rattlers. Some (but not all) of these proxy methods apply this criterion recursively, thus more closely approximating the true stable network. However, this method cannot account for the presence of particles with at least $d+1$ stable contacting neighbors which are nevertheless not geometrically constrained.

Here, we present an alternative scheme for identifying rattlers that is intuitive, efficient, and physically meaningful. In fact, we have been using it for some time without realizing that it was not yet present in the literature~\cite{
corwin_bond_2013,
morse_geometric_2014, 
charbonneau_jamming_2015, 
charbonneau_universal_2016,
morse_geometric_2016, 
morse_hidden_2016, 
morse_echoes_2017, 
charbonneau_glassy_2019,
hagh_broader_2019, 
sartor_direct_2020,
dennis_jamming_2020,
morse_differences_2020, 
richard_predicting_2020, 
morse_direct_2021, 
rissone_longrange_2021,
sartor_meanfield_2021,
charbonneau_memory_2021,
charbonneau_finitesize_2021,
hagh_transient_2022,
stanifer_avalanche_2022, 
sartor_predicting_2022, 
dennis_emergence_2022,
charbonneau_jamming_2023}. 
Our method is based on a fundamental link between local rigidity and the local geometry of force carrying contacts, and implemented through the computation of the convex hull of the set of contacting particles. The stable network obtained by this algorithm is identical to that found in Refs.~\cite{donev_linear_2004, lerner_simulations_2013}.

The central thrust of our algorithm is based on a comment within Ref.~\cite{donev_linear_2004}, namely that a sphere can only be locally rigid if it has greater than $d+1$ non-cohemispheric contacts. While the authors of Ref.~\cite{donev_linear_2004} note that simple constructions can be done in low spatial dimensions (a method adopted in Refs.~\cite{wentworth-nice_structured_2020, zhang_jammed_2022}), ours is a dimensionally independent construction: a particle whose center is $\mathbf{r}_0$ is locally stable if the sum of all forces acting on it is zero, and if the surface of the convex hull of the particle's center and the centers of all of its contacting neighbors $\{\mathbf{r}_i\}$ does not include $\mathbf{r}_0$, i.e. $\mathbf{r}_0 \notin \partial \mathrm{Conv}(\mathbf{r}_0, \{\mathbf{r}_i\})$, where $\partial \mathrm{Conv}$ is the surface of the convex hull. We also prove a related theorem, which can be shown to be equivalent to this, which states that a particle is locally stable if the maximum inscribed sphere of its radical Voronoi cell is unique and identical to the particle itself.

The rest of this article is structured as follows. In Sec.~\ref{sec:definitions}, we provide definitions for the generalized packing models that we can consider and a series of mathematical definitions which will allow us to prove the two main theorems. In Sec.~\ref{sec:proofs}, we provide a formal proof that each construction finds the correct stable network. In Sec.~\ref{sec:complexity}, we address computational complexity, noting that even in the worst case scenario, the convex hull algorithm is faster than the linear programming algorithm in $d<6$. We conclude in Sec.~\ref{sec:conclusion} by discussing extensions of this construction to other models.

\section{Definitions}
\label{sec:definitions}

In the following, bold letters denote vectors in $\mathbb{R}^d$, $\mathbf{0}$ represents the zero-vector, $\mathbf{a}\cdot\mathbf{b}$ denotes the dot product between vectors $\mathbf{a}$ and $\mathbf{b}$, $\{\mathbf{r}_i\}$ denotes a finite set of points, where each point is represented by a vector from the origin, and $\{\mathbf{r}_i\} \setminus \mathbf{r}_0$ denotes the set $\{\mathbf{r}_i\}$ excluding the point $\mathbf{r}_0$. All definitions assume the standard Euclidean distance metric on $\mathbb{R}^d$, where the distance between points $\mathbf{a}$ and $\mathbf{b}$ is denoted $|\mathbf{a}-\mathbf{b}|$. To define our packing, and to aid in later definitions and theorems, we define both open and closed balls.

\begin{definition}
An open ball of radius $\sigma$ around $\mathbf{s}$ is defined as the set of points contained within a distance $\sigma$ of $\mathbf{s}$. The notation we will use is ${B_\sigma (\mathbf{s}) \equiv \{\mathbf{y} : |\mathbf{s}-\mathbf{y}| < \sigma\}}$.
\end{definition}

\begin{definition}
A closed ball of radius $\sigma$ around $\mathbf{s}$ is defined as the set of points contained within and including a distance $\sigma$ of $\mathbf{s}$. The notation we will use is ${\overline{B}_\sigma (\mathbf{s}) \equiv \{\mathbf{y} : |\mathbf{s}-\mathbf{y}| < \sigma\}}$.
\end{definition}

We thus consider particles defined by $\overline{B}_{\sigma_i}(\mathbf{r}_i)$ with a non-dimensional overlap between particles $i$ and $j$ defined as
\begin{equation}
h_{ij} \equiv 1-\frac{\abs{\mathbf{r}_i - \mathbf{r}_j}}{\sigma_i + \sigma_j},
\end{equation}
subject to an additive potential $U = \sum_{ij}u(h_{ij})$ where contacts ($h_{ij}\ge0)$ coincide with the potential cutoff, i.e. $u(h_{ij}\le0)=0$. This form includes (but is not limited to) standard soft-sphere contact power law potentials where $u(h_{ij}>0) \propto h_{ij}^\gamma$ for $\gamma > 0$ ($\gamma = 2$ for Hookean spheres, and $\gamma = 2.5$ for Hertzian spheres) and hard spheres, where $u(h_{ij}>0) = \infty$. 

From this, the force on particle $i$ from particle $j$ can be defined as 
\begin{equation}
\mathbf{f}_{ij} \equiv \nabla u(h_{ij}) = \abs{\nabla u(h_{ij})}\frac{\mathbf{r}_j-\mathbf{r}_i}{\abs{\mathbf{r}_j-\mathbf{r}_i}}.
\end{equation}
Here the only salient feature is that the force points towards the particle center from the point of contact. Unless otherwise mentioned, we consider only packings which are in a local energy minimum, such that the sum of forces acting on each particle is zero. Extensions to non energy minimized packings will be considered in Sec.~\ref{sec:conclusion}.

\begin{definition}[Adapted from Ref.~\cite{donev_linear_2004}]
A particle is locally stable if the sum of the forces acting on it is zero and the forces acting on it span $\mathbb{R}^d$. Particles which are not locally stable are called unstable.
\label{def:stable}
\end{definition}

In an effort to make this work as self contained as possible, we have compiled a list of the mathematical definitions necessary to follow the theorems and proofs of Sec.~\ref{sec:proofs} such that only basic knowledge of set theory and linear algebra will be prerequisite. The definitions are adapted from Refs.~\cite{ziegler_lectures_1995, munkres_topology_2000, grunbaum_convex_2003}.

\begin{figure}[htp!]
\includegraphics[width=0.5\linewidth]{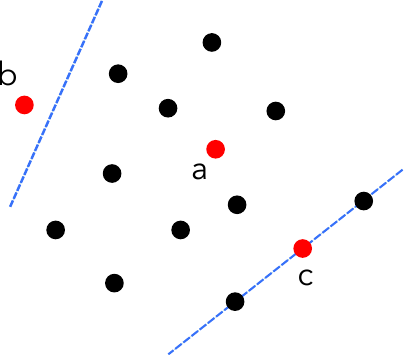}
\caption{Here we demonstrate the concept of an extreme point by examining three red particles labelled (a-c). While this example is embedded in $d=2$, the demonstration extends naturally to higher dimensions, replacing lines with $(d-1)$-planes. (a) No line can be drawn which separates the particle from all other particles, so (a) is not an extreme point. (b) A line can be drawn which separates the particle from all other points, and it is thus an extreme point and will be shown to be on the surface of the convex hull. (c) No line can be drawn which separates the particle from all other particles, so it is not an extreme point. However, a line exists which contains the particle and which divides space such that all particles exist (inclusively) in one of its half spaces, thus the point is on the surface of the convex hull.}
\label{fig:extremalPoints}
\end{figure}

\begin{definition}
An extreme point $\mathbf{r}_0$ of the finite set $\{\mathbf{r}_i\}$ is a point which can be separated from all other points by a $(d-1)$-plane. Thus there exists a vector $\mathbf{a}\in\mathbb{R}^d$ with at least one non-zero element and $b\in\mathbb{R}$ for which ${\mathbf{a}\cdot\mathbf{r}_0 - b > 0}$ while ${\mathbf{a}\cdot\mathbf{r}_j - b \le 0}$ for all ${\mathbf{r}_j \in \{\mathbf{r}_i\} \setminus \mathbf{r}_0}$. An illustration of both extreme and non-extreme points is given in Fig.~\ref{fig:extremalPoints}. \label{def:extremePoint}
\end{definition}

\noindent Remark: In our proofs, we only need the extreme points of finite sets. The concept of an extreme point can of course be generalized to infinite sets~\cite{munkres_topology_2000}, but this makes several of the theorems unwieldy. The definition used here is non-standard but reduces to the common definition in the case of finite sets.

\begin{figure}[htp!]
\includegraphics[width=0.5\linewidth]{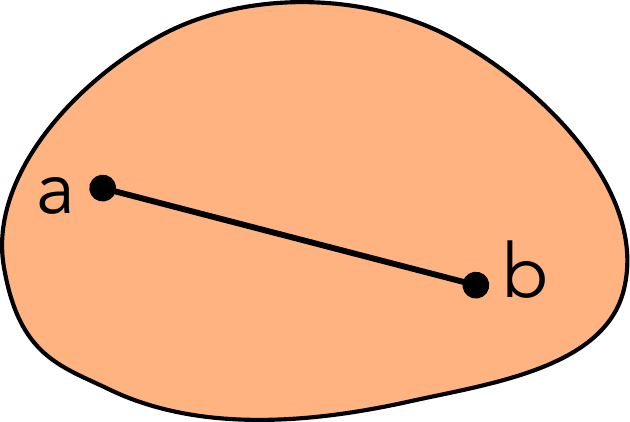}
\caption{A simple illustration that a convex set containing points $\mathbf{a}$ and $\mathbf{b}$ contains all points on a straight line between them.}
\label{fig:convexPoints}
\end{figure}

\begin{definition}
A set $K \subset \mathbb{R}^d$ is convex if for all $\mathbf{a}, \mathbf{b} \in K$,  $\mathbf{c} = (t-1)\mathbf{a} + t\mathbf{b} \in K$ for all $t \in [0,1]$. Put simply, if $\mathbf{a}$ and $\mathbf{b}$ are in $K$, then $K$ is convex if every point $\mathbf{c}$ along the straight line between $\mathbf{a}$ and $\mathbf{b}$ is also in $K$. This is illustrated in Fig.~\ref{fig:convexPoints}.
\label{def:convex}
\end{definition}

\begin{definition}
From Ref.~\cite{grunbaum_convex_2003}, a compact convex set ${K \subset \mathbb{R}^d}$ is a convex polytope if the extreme points of $K$ form a finite set. In this work, all instances of the word polytope are implied to be convex.
\end{definition}

\begin{definition}
The surface $\partial K$ of a polytope $K$ is defined as the infinite set of points $\mathbf{s} \in K$ for which there exists $\mathbf{s}_\mathrm{out} \in B_\sigma (\mathbf{s})$ where $\mathbf{s}_\mathrm{out} \notin K$ for all $\sigma$.
\end{definition}

\begin{figure}[htp!]
\includegraphics[width=0.5\linewidth]{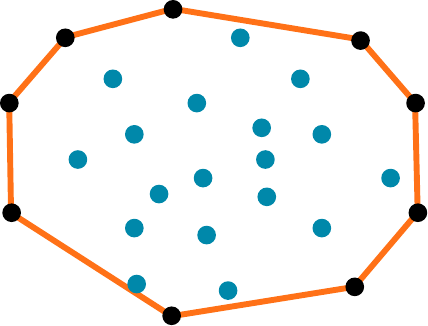}
\caption{Here we demonstrate the convex hull (orange) of a set of points. Points on the surface of the convex hull are colored black, while points not on the surface of the convex hull are in teal.}
\label{fig:convHullExample}
\end{figure}

\begin{definition}
The convex hull of a set of points $\mathrm{Conv}(\{\mathbf{r}_i\})$ is the unique closed $d$-dimensional polytope containing all points $\{\mathbf{r}_i\}$ whose vertices are members of $\{\mathbf{r}_i\}$. The surface of the convex hull is denoted $\partial \mathrm{Conv}(\{\mathbf{r}_i\})$ and is shown visually in Fig.~\ref{fig:convHullExample}. 
\end{definition}

\begin{definition}
For a sphere given by $\overline{B}_\sigma(\mathbf{r})$, the points $\{\mathbf{b}_i\} \subset \partial \overline{B}_\sigma(\mathbf{r})$ are cohemispheric if there exists $\mathbf{a} \in \mathbb{R}^d$ with at least one non-zero element, where $\mathbf{a} \cdot (\mathbf{b}_i - \mathbf{r}) \ge 0$ for all $i$. Similarly, forces $\{\mathbf{f}_i\}$ are cohemispheric if $\mathbf{a} \cdot \mathbf{f}_i \ge 0$ for all $i$. If no such $\mathbf{a}$ exists, the points or forces are non-cohemispheric. 
\end{definition}

\begin{figure}[htp!]
\includegraphics[width=0.9\linewidth]{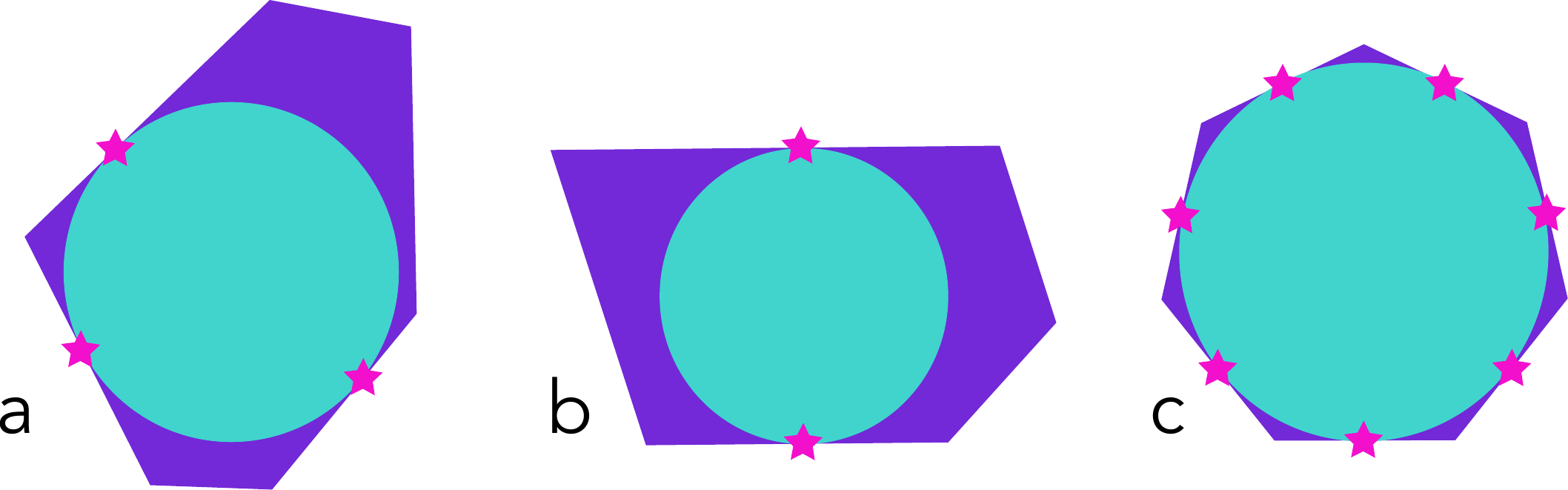}
\caption{The maximum inscribed sphere (teal) of a convex polytope (purple) in $d=2$. Contact points between the MIS and the polytope are shown with stars. (a) The generic case with no symmetries has $d+1$ contact points between the polytope and the MIS. (b) When two of the contacting surfaces are parallel, it is possible to have an MIS with only 2 contacts in any dimension. (c) Highly symmetric polytopes (regular ones, as shown here, or those near jamming), may have MIS which have greater than $d+1$ contacts with the polytope.}
\label{fig:misdiagram}
\end{figure}

\begin{definition}
For a polytope $K$, the maximum inscribed sphere $M(K)$ is the largest closed ball fully contained in $K$. That is, $M(K) = \mathrm{max}_\sigma [\overline{B}_\sigma(\mathbf{r}): \overline{B}_\sigma(\mathbf{r}) \subset K]$. An illustration of the concept, including generic, degenerate, and highly symmetric cases is given in Fig.~\ref{fig:misdiagram}. We use MIS as an abbreviation when not referring to a specific $M(K)$.
\end{definition}

\begin{definition}
In a packing of particles with positions $\{\mathbf{r}_i\}$, the Voronoi cell of particle 0 is the set $V(\mathbf{r}_0) = \{\mathbf{y} : |\mathbf{y} - \mathbf{r}_0| \le |\mathbf{y} - \mathbf{r}_i| \,\,\forall i\}$.
\end{definition}

\begin{definition}
The power of a point $\mathbf{c} \in \mathbb{R}^d$ with respect to a sphere with center $\mathbf{r}$ and radius $\sigma$ is given by $\Pi_{\mathbf{r},\sigma}(\mathbf{c}) = |\mathbf{r} - \mathbf{c}|^2 - \sigma^2$. Points on the interior of the sphere have negative power, points on the surface of the sphere have zero power, and points outside of the sphere have positive power.
\label{def:power}
\end{definition}

\begin{definition}
In a packing of particles with positions $\{\mathbf{r}_i\}$ and radii $\sigma_i$, the radical Voronoi cell of particle 0 is the set $R(\mathbf{r}_0) = \{\mathbf{y} : \Pi_{\mathbf{r}_0,\sigma_0}(\mathbf{y}) \le \Pi_{\mathbf{r}_i,\sigma_i}(\mathbf{y}) \,\,\forall i\}$.
\end{definition}

Trivially, we see that if all particles are the same size (i.e. $\sigma_i = \sigma$ for all $i$), then the radical Voronoi cell reduces to that of the standard Voronoi cell. Both the radical Voronoi cell and, by extension, the Voronoi cell are convex polytopes, and it is from the definitions that these cells tessellate space, i.e. there is no point in space which is not contained in the radical Voronoi cell of a particle, and the only points which can be contained in multiple radical Voronoi cells are on the shared surfaces of two or more cells.

\section{Proofs of the stability theorems}
\label{sec:proofs}

In this section, we provide proofs of the two main stability theorems, labelled Theorem~\ref{theor:mainCH} (Sec.~\ref{sec:chProof}) and Theorem~\ref{theor:mainRV} (Sec.~\ref{sec:rvProof}). While Sec.~\ref{sec:chProof} is entirely self contained, Sec.~\ref{sec:rvProof} uses theorems from Sec.~\ref{sec:chProof}. Some of the theorems are elementary or have been proven by simpler means elsewhere, but we formulate our own versions here, as we believe that they help to build the physical intuition for the main theorems.

\subsection{Stability via the convex hull}
\label{sec:chProof}

\begin{theorem}
[The Krein-Milman theorem~\cite{krein_extreme_1940}]
A compact convex subset of a Hausdorff locally convex topological vector space is equal to the closed convex hull of its extreme points.
\end{theorem}

The proof of this theorem is given in Ref.~\cite{krein_extreme_1940}. For the purposes of this work, we will use the fact that the standard vector space on $\mathbb{R}^d$ with the Euclidean distance metric and standard inner product is a Hausdorff locally convex topological vector space. For clarification of these terms, we suggest any standard textbook on topology (for example, Ref.~\cite{munkres_topology_2000}).

\begin{corollary}
The convex hull of a set of points $\mathrm{Conv}(\{\mathbf{r}_i\})$ is equal to the closed $d$-dimensional polytope whose vertices are the extreme points of $\{\mathbf{r}_i\}$.
\label{cor:convVertPoly}
\end{corollary}

\begin{proof}
Given that the standard vector space on $\mathbb{R}^d$ is a Hausdorff locally convex topological vector space, the Krein-Milman theorem states that a closed convex polytope is the convex hull of its extreme points, which for a convex polytope are its vertices.
\end{proof}

\begin{corollary}
If $\mathbf{r}_0$ is an extreme point of $\{\mathbf{r}_i\}$ then $\mathbf{r}_0 \in \partial \mathrm{Conv}(\{\mathbf{r}_i\})$.
\label{cor:exOnHull}
\end{corollary}

\begin{proof} 
To prove that $\mathbf{r}_0 \in \partial \mathrm{Conv}(\{\mathbf{r}_i\})$, we must show that there exists a point $\mathbf{s}_\mathrm{out} \in B_\sigma(\mathbf{r}_0)$ such that $\mathbf{s}_\mathrm{out} \notin \mathrm{Conv}(\{\mathbf{r}_i\})$. Because $\mathbf{r}_0$ is an extreme point, there exists $\mathbf{a} \in \mathbb{R}^d$ with at least one non-zero element and $b \in \mathbb{R}$ such that $\mathbf{a} \cdot \mathbf{r}_0 - b > 0$ while $\mathbf{a} \cdot \mathbf{r}_j - b \le 0$ for all $\mathbf{r}_j \in \{\mathbf{r}_i\} \setminus \mathbf{r}_0$. We can thus construct $\mathbf{s}_\mathrm{out} = \mathbf{r}_0 + \frac{\sigma\mathbf{a}}{2\abs{\mathbf{a}}}$, for which $|\mathbf{r}_0 - \mathbf{s}_\mathrm{out}| = \frac{\sigma}{2}$, and thus $\mathbf{s}_\mathrm{out} \in B_\sigma(\mathbf{r}_0)$. By construction, $\mathbf{s}_\mathrm{out}$ is an extreme point of the set $\{\mathbf{s}_\mathrm{out}, \mathbf{r}_i\}$, and thus by Corollary~\ref{cor:convVertPoly}, $\mathbf{s}_\mathrm{out} \notin \mathrm{Conv}(\mathbf{r}_i)$. This statement is true for any value of $\sigma$, and thus $\mathbf{r}_0 \in \partial \mathrm{Conv}(\{\mathbf{r}_i\})$.
\end{proof}

\begin{theorem}
A full dimensional convex polytope is equivalently defined by either its vertices (V-Representation) or the intersection of half-planes representing its surface (H-Representation). 
\label{theor:hrepvrep}
\end{theorem}

The proof of this theorem is contained in standard texts on convex polytopes, for example following the proofs of Theorems 3.1.1 and 3.1.2 of Ref.~\cite{grunbaum_convex_2003} or Theorem 1.1 of Ref.~\cite{ziegler_lectures_1995}. The theorem only applies to full dimensional polytopes (i.e. ones which are $d$-dimensional objects), but if the polytope is a $d'$ dimensional object, where $d'\neq d$, it is sufficient for our purposes to consider the V-Representation and the H-Representation in $\mathbb{R}^{d'}$, in which the polytope is full dimensional.

\begin{corollary}
A point $\mathbf{r}_0$ which is contained on a $(d-1)$-plane, which defines a halfspace containing all $\mathbf{r}_i$ is contained in $\partial \mathrm{Conv}(\{\mathbf{r}_i\})$. That is, if there exists $\mathbf{a}\in\mathbb{R}^d$ with at least one non-zero element and $b\in\mathbb{R}$ such that $\mathbf{a}\cdot\mathbf{r}_0 - b = 0$ and $\mathbf{a}\cdot\mathbf{r}_j - b \le 0$ for all $\mathbf{r}_j \in \{\mathbf{r}_i\} \setminus \mathbf{r}_0$, then $\mathbf{r}_0 \in \partial \mathrm{Conv}(\{\mathbf{r}_i\})$
\end{corollary}

\begin{proof} 
There are two cases here which need to be proven. If $\mathbf{r}_0$ is an extreme point, then $\mathbf{r}_0 \in \partial \mathrm{Conv}(\{\mathbf{r}_i\})$ by Corollary
~\ref{cor:exOnHull}. If $\mathbf{r}_0$ is not an extreme point, then the half-plane representation described here is equivalent to that defining the H-Representation of a convex polytope, and thus $\mathbf{r}_0 \in \mathrm{Conv}(\{\mathbf{r}_i\})$ by Theorem~\ref{theor:hrepvrep}. The further statement that $\mathbf{r}_0 \in \partial \mathrm{Conv}(\{\mathbf{r}_i\})$ comes directly from the definition of a halfspace. 
\end{proof}

\begin{theorem}
Any set of $d$ or fewer points on the surface of a sphere are cohemispheric. That is, for a sphere centered at $\mathbf{r}_0$ with radius $\sigma_0$ and points $\{\mathbf{c}_i\}$ satisfying $|\mathbf{c}_i - \mathbf{r}_0| = \sigma_0$, there exists a vector $\mathbf{a}\in\mathbb{R}^d$ such that $\mathbf{a} \cdot (\mathbf{c}_i - \mathbf{r}_0) \ge 0$ for all $i$.
\label{theor:nonCohemisphericPoints}
\end{theorem}

\begin{proof}
Here, we can relax the condition $|\mathbf{c}_i - \mathbf{r}_0| = \sigma_0$ and prove a more general theorem. A hyperplane in $\mathbb{R}^d$ can always be formed which passes through the $d$ contact points. That is, there exist $\mathbf{a}' \in \mathbb{R}^d$ and $b \in \mathbb{R}^d$ such that $\mathbf{a}' \cdot \mathbf{c}_i = b$ for all $\mathbf{c}_i$. Note that if we construct a matrix $C$ with rows $\mathbf{c}_i$, then this hyperplane is not unique if $\mathrm{det}(C)=0$, but any of the infinitely many solutions will suffice.

We can define $b' \in \mathbb{R}^d$ by $\mathbf{a}' \cdot \mathbf{r}_0 = b'$, then $\mathbf{a}' \cdot (\mathbf{c}_i - \mathbf{r}_0) = b - b'$. If $b' \le b$, then  $b-b' \ge 0$, and we can take $\mathbf{a} = \mathbf{a}'$, whereupon the theorem is proven. If $b' > b$, then we can take $\mathbf{a} = -\mathbf{a}'$, whereupon the theorem is proven.
\end{proof}

From this, we note that the minimal number of points $\mathbf{c}_i$ for which this theorem no longer holds is $d+1$. This is not to say that \textit{any} $d+1$ points on the surface are non-cohemispheric (see for example, Fig.~\ref{fig:rattlerDef}a), but to state that \textit{the minimal} number of points on a sphere which are non-cohemispheric is $d+1$.

\begin{theorem}
Given $\mathbf{f}_{i1} \in \{\mathbf{f}_i\}$ and $\mathbf{a} \in \mathbb{R}^d$ where $\mathbf{f}_{i1} \neq \mathbf{0}$, $\mathbf{a} \neq \mathbf{0}$, and $\mathbf{a}\cdot\mathbf{f}_{i1} \neq 0$, $\sum_i\mathbf{f}_i = \mathbf{0}$ only if there exists $\mathbf{f}_{i2} \in \{\mathbf{f}_i\} \setminus \mathbf{f}_{i1}$ such that $\mathrm{sign}(\mathbf{a} \cdot \mathbf{f}_{i1}) = -\mathrm{sign}(\mathbf{a} \cdot \mathbf{f}_{i2}) $.
\label{theor:zeroSum}
\end{theorem}

\begin{proof}
Here we project $\mathbf{a}$ onto the sum yielding ${\sum_i(\mathbf{a} \cdot \mathbf{f}_{i}) = \mathbf{a} \cdot \mathbf{f}_{i1} + \sum_{i\neq i1}(\mathbf{a} \cdot \mathbf{f}_{i}) = 0}$. This last equality can only be true if there is at least one element of the sum which is of the opposite sign of $\mathbf{a} \cdot \mathbf{f}_{i1}$, implying that there exists $\mathbf{f}_{i2} \in \{\mathbf{f}_i\} \setminus \mathbf{f}_{i1}$ such that $\mathrm{sign}(\mathbf{a} \cdot \mathbf{f}_{i1}) = -\mathrm{sign}(\mathbf{a} \cdot \mathbf{f}_{i2}) $.
\end{proof}

This theorem is meant to be a vector extension of the trivial theorem that a sum of numbers can only be zero if either all elements are zero, or if it contains both positive and negative elements. Setting aside the null case, this theorem simply states that a sum of vectors with at least one non-zero element can only be zero if it contains positive and negative elements when projected onto (almost) any axis. A mild caveat must be added, namely that the projection is not onto a vector normal to a chosen non-zero vector in the set. This caveat is only a formality as the projecting vector $\mathbf{a}$ is arbitrary.

\begin{corollary}
A particle with zero net force and at least $d+1$ non-cohemispheric non-zero forces is locally stable.
\label{cor:nonCohemisphericForces}
\end{corollary}

\noindent Remark: 
This theorem applies more generally to both point particles and any shape of particle with forces pointing towards its center of mass. Such a particle will be stable to translations, but not to rotations.

\begin{proof}

We label the set of non-zero forces $\{\mathbf{f}_i\}$ and the particle center by $\mathbf{r}$. Note that the minimum number of vectors needed to span $\mathbb{R}^d$ is $d$, so a particle is unstable with fewer than $d$ forces acting upon it. Furthermore, a particle with $d$ forces acting upon it is unstable by Theorem~\ref{theor:nonCohemisphericPoints}, as these forces are necessarily cohemispheric, and thus there exists $\mathbf{a} \in \mathbb{R}^d$ such that $\mathbf{a} \cdot \mathbf{f_i} \ge 0$ for all $i$. By Theorem~\ref{theor:zeroSum}, $\sum \mathbf{f}_i \neq 0$ unless $\mathbf{f}_i = 0$ for all $i$, and thus a particle with $d$ non-zero forces acting upon it is unstable.

By definition, if there are $d+1$ non-cohemispheric non-zero forces, then no $\mathbf{a}$ exists for which $\mathbf{a} \cdot \mathbf{f}_i \ge 0$ for all $i$. Thus, for all $\mathbf{a} \in \mathbb{R}^d$ with at least one non-zero element, Theorem~\ref{theor:zeroSum} states that there will be positive and negative projections, and thus the net force can sum to zero without all forces being trivially zero, and thus the particle is locally stable.
\end{proof}

\begin{theorem}
A particle with center $\mathbf{r}_0$ and with contacting particles centered at $\{\mathbf{r}_i\}$ is unstable if $\mathbf{r}_0 \in \partial \mathrm{Conv}(\mathbf{r}_0,\{\mathbf{r}_i\})$.
\label{theor:unstableHull}
\end{theorem}

\begin{proof}
We have two instances to prove. If $\mathbf{r}_0$ is an extreme point, then by Definition~\ref{def:extremePoint}, there exist $\mathbf{a} \in \mathbb{R}^d$ and $b \in \mathbb{R}$ where $\mathbf{a} \cdot \mathbf{r}_0 - b > 0$ while $\mathbf{a} \cdot \mathbf{r}_j - b \le 0$ for all $\mathbf{r}_j \in \{\mathbf{r}_i\} \setminus \mathbf{r}_0$. The contact forces on $\mathbf{r}_0$ are all of the form $\mathbf{f}_j = c_j(\mathbf{r}_0 - \mathbf{r}_j)$ with $c_j \in \mathbb{R}$ and $c_j \ge 0$. Thus $\sum_j \mathbf{f}_j = \sum_j c_j(\mathbf{r}_0 - \mathbf{r}_j)$. Taking the projection on $\mathbf{a}$, we have $\sum_j \mathbf{a} \cdot \mathbf{f}_j = \sum_j c_j(\mathbf{a} \cdot \mathbf{r}_0 - \mathbf{a} \cdot \mathbf{r}_j)$. Depending on the sign of $b$, the non-zero terms are either all positive or all negative, meaning that the sum cannot be zero unless all $c_j$ are zero. Thus, by Theorem~\ref{theor:zeroSum}, either $\sum_j \mathbf{f}_j \neq \mathbf{0}$, or $\mathbf{f}_j = 0$ for all $j$. Either condition means that the particle is unstable. 

If $\mathbf{r}_0$ were not an extreme point, then the sum $\sum_j \mathbf{a} \cdot \mathbf{f}_j$ could only be $0$ if $\mathbf{a} \cdot \mathbf{r}_j - b = 0$ for all $j$. These forces would then all be co-hemispheric, and thus by Theorem~\ref{theor:zeroSum}, either $\sum_j \mathbf{f}_j \neq \mathbf{0}$, or $\mathbf{f}_j = 0$ for all $j$, and thus the particle is unstable.
\end{proof}

Here we note that this is a sufficient condition for instability, and not a necessary one. If $\mathbf{r}_0$ is out of force balance with neighboring contacts $\{\mathbf{r}_i\}$, but $\mathbf{r}_0 \notin \mathrm{Conv}(\mathbf{r}_0,\{\mathbf{r}_i\})$, then $\mathbf{r}_0$ is still unstable.

\begin{figure}[htp!]
\includegraphics[width=0.7\linewidth]{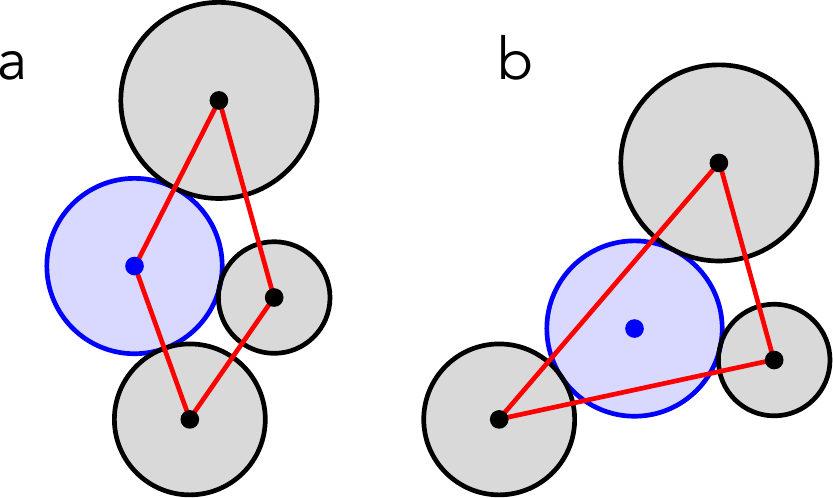}
\caption{We test whether the blue particle is stable by looking at the convex hull (red) of its own center and the centers of its stable neighboring particles (in black). Note that there may be other contacts with the blue particle which have been determined to be unstable and are thus not shown. In (a) the blue particle is unstable, because its center lies on the surface of the convex hull. In (b) the blue particle is stable, because its center is not on the surface of the convex hull.}
\label{fig:rattlerDef}
\end{figure}

\begin{theorem}
A particle with center  $\mathbf{r}_0$ and with a non-empty set of stable contacting particles centered at $\{\mathbf{r}_i\}$ is locally stable if and only if $\mathbf{r}_0 \notin \partial \mathrm{Conv}(\mathbf{r}_0,\{\mathbf{r}_i\})$ and the sum of forces acting on the particle is zero.
\label{theor:mainCH}
\end{theorem}

\begin{proof}
The statement that $\mathbf{r}_0$ is locally stable if $\mathbf{r}_0 \notin \partial \mathrm{Conv}(\mathbf{r}_0,\{\mathbf{r}_i\})$, and the sum of all forces acting on the particle from $\{\mathbf{r}_i\}$ is zero follows a recursive application of Definition~\ref{def:stable} and Theorem~\ref{theor:unstableHull}. 

Next, we must prove that $\mathbf{r}_0 \notin \partial \mathrm{Conv}(\mathbf{r}_0,\{\mathbf{r}_i\})$ with stable contacts $\{\mathbf{r}_i\}$ and zero net force implies that $\mathbf{r}_0$ is locally stable and thus has a set of stable forces acting on the particle centered at $\mathbf{r}_0$ which both span $\mathbb{R}^d$ and sum to zero. Because $\mathbf{r}_0 \notin \partial \mathrm{Conv}(\mathbf{r}_0,\{\mathbf{r}_i\})$, we know that $\mathbf{r}_0$ is neither an extreme point of the convex hull, nor is it on the surface. Thus no $\mathbf{a}$ exists for which the contact forces, labelled $\{\mathbf{f}_i\}$ have the property $\mathbf{a} \cdot \mathbf{f}_i \ge 0$ for all $i$. These forces are thus non-cohemispheric, and so from Theorem~\ref{theor:nonCohemisphericPoints}, there must be $d+1$ of them. And because this particle has zero net force acting upon it, by Corollary~\ref{cor:nonCohemisphericForces}, the particle is locally stable.

An illustration of this theorem is given in Fig.~\ref{fig:rattlerDef}.
\end{proof}

\subsection{Stability via the radical Voronoi cell}
\label{sec:rvProof}

\begin{theorem}
If $i$ and $j$ are hard particles with centers $\mathbf{r}_i$ and $\mathbf{r}_j$ and radii $\sigma_i$ and $\sigma_j$ and $h_{ij} = 0$, then $\overline{B}_{\sigma_i}(\mathbf{r}_i) \cap \overline{B}_{\sigma_j}(\mathbf{r}_j)$ contains exactly one point $\mathbf{c}_{ij}$ where $\mathbf{c}_{ij} \in \partial R(\mathbf{r}_i)$ and $\mathbf{c}_{ij} \in \partial R(\mathbf{r}_j)$.
\label{theor:HStouchRV}
\end{theorem}

\begin{figure}[htp!]
\includegraphics[width=0.7\linewidth]{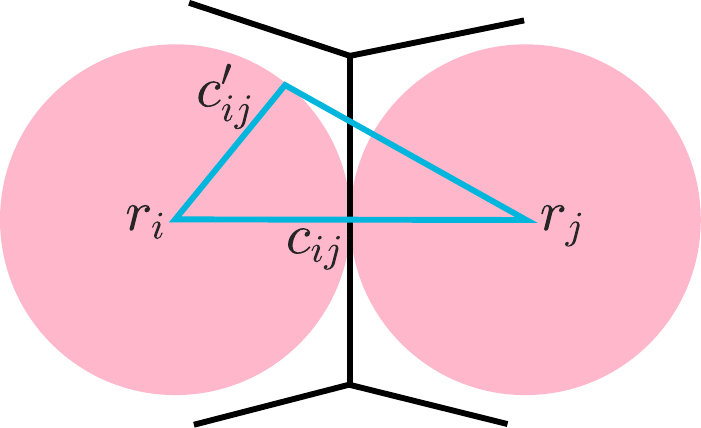}
\caption{The radical Voronoi diagram is shown between contacting particles $i$ and $j$ with contact point $\mathbf{c}_{ij}$. A second contact point between the two particles $\mathbf{c}_{ij}'$ is assumed, so that we can show $\mathbf{c}_{ij}' = \mathbf{c}_{ij}$ via the triangle inequality.}
\label{fig:HSRadVoroContact}
\end{figure}

\begin{proof}
We define 
\begin{equation}
\mathbf{c}_{ij} = \mathbf{r}_i + \sigma_i \frac{\mathbf{r}_j - \mathbf{r}_i}{|\mathbf{r}_j - \mathbf{r}_i|}
\label{eq:contactPoints}
\end{equation}
and note that $|\mathbf{c}_{ij} - \mathbf{r}_i| = \sigma_i$ so that $\mathbf{c}_{ij} \in \overline{B}_{\sigma_i}(\mathbf{r}_i)$ and $|\mathbf{c}_{ij} - \mathbf{r}_j| = |(\mathbf{r}_j - \mathbf{r}_i)|-\sigma_i$. We then note that $h_{ij} = 0$ implies $\sigma_j = |(\mathbf{r}_j - \mathbf{r}_i)|-\sigma_i$, and thus $|\mathbf{c}_{ij} - \mathbf{r}_j| = \sigma_j$ and so $\mathbf{c}_{ij} \in \overline{B}_{\sigma_j}(\mathbf{r}_j)$. 

To show that the intersection contains only one point, we assume that $\mathbf{c}_{ij}' \in \overline{B}_{\sigma_i}(\mathbf{r}_i) \cap \overline{B}_{\sigma_j}(\mathbf{r}_j)$ so that $|\mathbf{c}'_{ij} - \mathbf{r}_i| \le \sigma_i$ and $|\mathbf{c}'_{ij} - \mathbf{r}_j| \le \sigma_j$, but $\mathbf{c}_{ij}' \neq \mathbf{c}_{ij}$ (as in Fig.~\ref{fig:HSRadVoroContact}). By the triangle inequality, $|\mathbf{r}_j - \mathbf{r}_i| \le |\mathbf{r}_i - \mathbf{c}_{ij}'| + |\mathbf{r}_j - \mathbf{c}_{ij}'|$, which becomes the degenerate statement $\sigma_i + \sigma_j \le \sigma_i + \sigma_j$. The degeneracy implies a triangle of zero area, so that $\mathbf{c}_{ij}'$ lies on the line between $\mathbf{r}_i$ and $\mathbf{r}_j$, and by simple algebra, we find that $\mathbf{c}_{ij}' = \mathbf{c}_{ij}$. This is a contradiction, and thus the intersection $\overline{B}_{\sigma_i}(\mathbf{r}_i) \cap \overline{B}_{\sigma_j}(\mathbf{r}_j)$ contains only one point.

To show that $\mathbf{c}_{ij} \in R(\mathbf{r}_i)$ and $\mathbf{c}_{ij} \in R(\mathbf{r}_j)$, we calculate the power of $\mathbf{c}_{ij}$ with respect to each sphere. Here we find that $\Pi_{\mathbf{r}_i,\sigma_i}(\mathbf{c}_{ij}) = \Pi_{\mathbf{r}_j,\sigma_j}(\mathbf{c}_{ij}) = 0$. The only lower power would be negative (interior of a sphere), and because these are hard spheres, that is not possible. Thus, $\mathbf{c}_{ij} \in R(\mathbf{r}_i)$ and $\mathbf{c}_{ij} \in R(\mathbf{r}_j)$.
\end{proof}

\begin{corollary}
In a hard particle system, $\overline{B}_{\sigma_i}(\mathbf{r}_i) \cap \partial R(\mathbf{r}_i)$ contains only the contact points between particle $i$ and its contacting neighbors, centered at $\{\mathbf{r}_j\}$.
\label{cor:RVTangentPlane}
\end{corollary}

\begin{proof}
We know from Theorem~\ref{theor:HStouchRV} that $\overline{B}_{\sigma_i}(\mathbf{r}_i) \cap \partial R(\mathbf{r}_i)$ contains the contact points between particle $i$ and its contacting neighbors, so we need now only show that it contains no other points. Suppose $\mathbf{b} \in \overline{B}_{\sigma_i}(\mathbf{r}_i) \cap \partial R(\mathbf{r}_i)$ and that $\mathbf{b} \neq \mathbf{c}_{ij}$ from Eq.~\eqref{eq:contactPoints} for any $j$. Points on $\partial R(\mathbf{r}_i)$ have equal power with respect to at least one other sphere, which we will generically call $\mathbf{r}_j$. We have so far covered the case of zero power, and now consider points with negative power. As per Definition~\ref{def:power}, points of negative power are on the interior of both spheres, i.e. $\mathbf{b} \in B_{\sigma_i}(\mathbf{r}_i) \cap B_{\sigma_j}(\mathbf{r}_j)$, but because $i$ and $j$ are hard spheres $B_{\sigma_i}(\mathbf{r}_i) \cap B_{\sigma_j}(\mathbf{r}_j) = \emptyset$. Thus points of negative power are not in the intersection $\overline{B}_{\sigma_i}(\mathbf{r}_i) \cap \partial R(\mathbf{r}_i)$. Points of positive power are not contained within $\overline{B}_{\sigma_i}(\mathbf{r}_i)$ and are thus not in the intersection $\overline{B}_{\sigma_i}(\mathbf{r}_i) \cap \partial R(\mathbf{r}_i)$. Therefore, $\overline{B}_{\sigma_i}(\mathbf{r}_i) \cap \partial R(\mathbf{r}_i)$ contains only the contact points between particle $i$ and its contacting neighbors, centered at $\{\mathbf{r}_j\}$.
\end{proof}

\begin{theorem}
In a convex region $K$, if $\overline{B}_\sigma(\mathbf{a}) \subset K$ and $\overline{B}_\sigma(\mathbf{b}) \subset K$, then $\overline{B}_\sigma(\mathbf{c}) \subset K$ for all $\mathbf{c} = (t-1)\mathbf{a} + t\mathbf{b}$ where $t\in [0,1]$.
\label{theor:convexBall}
\end{theorem}

\begin{figure}[htp!]
\includegraphics[width=0.5\linewidth]{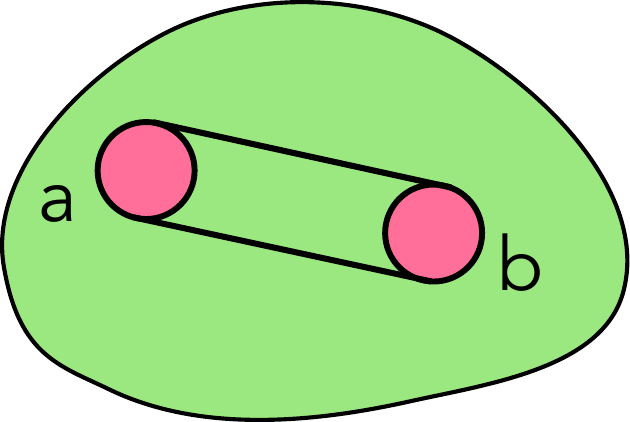}
\caption{An illustration of the fact that if two closed balls exist within a convex region (centered at $\mathbf{a}$ and $\mathbf{b}$ respectively), every closed ball on the line between the two is also contained in the region.}
\label{fig:convexBall}
\end{figure}

\begin{proof}
From Definition~\ref{def:convex}, this property is true for every individual point within the closed ball, so it is true for the closed ball itself. An illustration of the concept is given in Fig.~\ref{fig:convexBall}, where every ball contained on the line between $\mathbf{a}$ and $\mathbf{b}$ is contained in the convex region if the closed balls centered at $\mathbf{a}$ and $\mathbf{b}$ are contained in the region.
\end{proof}

\noindent Remark: We note that a further generalization of Theorem~\ref{theor:convexBall} is true when we have different radii balls at the endpoints $\overline{B}_{\sigma_a}(\mathbf{a})$ and $\overline{B}_{\sigma_b}(\mathbf{a})$, where then the interpolated ball has radius $\sigma_c = (t-1)\sigma_a + t\sigma_b$. This generalization is, however, not necessary for our purposes and would potentially obscure the results.

\begin{theorem}
If $M(R(\mathbf{r}_0))$ is not unique in a hard particle system, then the particle centered at $\mathbf{r}_0$ is not locally stable.
\label{theor:noUniqueNoStable}
\end{theorem}

\begin{figure}[htp!]
\includegraphics[width=0.7\linewidth]{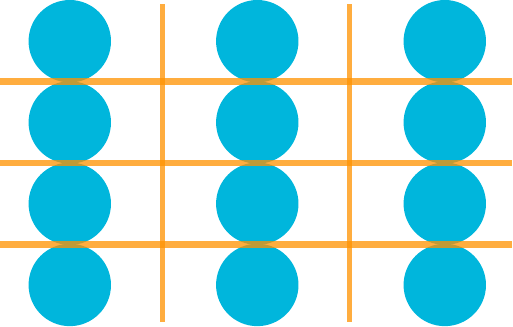}
\caption{An example of a radical Voronoi diagram (orange lines) for a set of particles (blue) which is highly degenerate. Here, because the MIS of each particle is not unique despite having radii equal to that of the particles, none of the particles are locally stable.}
\label{fig:radVoroDegen}
\end{figure}

\begin{proof}
We assume $M(R(\mathbf{r}_0))$ is not unique, such that $\overline{B}_\sigma(\mathbf{r}_1)\subset R(\mathbf{r}_0)$ and $\overline{B}_\sigma(\mathbf{r}_2) \subset R(\mathbf{r}_0)$ with $\mathbf{r}_1 \neq \mathbf{r}_2$ and there is no solution to $\overline{B}_\sigma'(\mathbf{r}_3) \subset R(\mathbf{r}_0)$ where $\sigma' > \sigma$. We then assume that the particle centered at $\mathbf{r}_0$ is locally stable and try to find a contradiction. If the particle is stable, there exist at least $d+1$ non-cohemispheric points $\mathbf{c}_{ij}$ given by Eq.~\eqref{eq:contactPoints} which, by Theorem~\ref{theor:HStouchRV}, have the property $\mathbf{c}_{ij} \in \overline{B}_{\sigma_0}(\mathbf{r}_0) \cap \partial R(\mathbf{r}_0)$. Because the particle centered at $\mathbf{r}_0$ is fully locally constrained, there exist no dilations or translations which maintain the hard sphere condition. We now have two scenarios to consider, which each contain a contradiction: $\sigma < \sigma_0$ and $\sigma \ge \sigma_0$.

If $\sigma < \sigma_0$, then neither $\overline{B}_\sigma(\mathbf{r}_1)$ nor $\overline{B}_\sigma(\mathbf{r}_1)$ represent the MIS, because $\overline{B}_{\sigma_0}(\mathbf{r}_0) \subset R(\mathbf{r}_0)$ has a larger radius. If $\sigma \ge \sigma_0$, then $\overline{B}_{\sigma_0}(\mathbf{r}_1) \subset \overline{B}_\sigma(\mathbf{r}_1) \subset R(\mathbf{\mathbf{r}_0})$. Theorem~\ref{theor:convexBall} states that all closed balls of radius $\sigma_0$ on the straight line between $\mathbf{r}_0$ and $\mathbf{r}_1$ are also contained in $R(\mathbf{r}_0)$. However, because the particle centered at $\mathbf{r}_0$ with radius $\sigma_0$ is stable, no translations $T$ exists such that $T(\overline{B}_{\sigma_0}(\mathbf{r}_0)) \subset R(\mathbf{r}_0)$. Because no case relating $\sigma$ and $\sigma_0$ exists without a contradiction, this implies that if $M(R(\mathbf{r}_0))$ is not unique in a hard particle system, then the particle centered at $\mathbf{r}_0$ is not locally stable.
\end{proof}

A packing with highly degenerate (non-unique) maximum inscribed spheres is illustrated in Fig.~\ref{fig:radVoroDegen}, where clearly the particles are not stable.

\begin{corollary}
If $M(K)$ is unique for a polytope $K$, then $M(K) \cap \partial K$ contains at least $d+1$ non-cohemispheric points.
\label{cor:uniqueMIS}
\end{corollary}

\begin{proof}
If $M(K)$ is unique, then there are no translations represented by the transformation $T$ which can be done such that $T(M(K)) \subset K$. Thus $M(K)$ is fully constrained by the boundary $\partial K$. By Corollary~\ref{cor:nonCohemisphericForces}, if we impose a fictive force on $M(K)$ from each point of contact $\{\mathbf{c}_i\}$ between $M(K)$ and $\partial K$, then there must be at least $d+1$ non-cohemispheric $\mathbf{c}_i$ for $M(K)$ to be stable. Thus $M(K) \cap \partial K$ contains at least $d+1$ non-cohemispheric points.
\end{proof}

\begin{theorem}
In a packing of hard particles, a particle with center  $\mathbf{r}_0$ and radius $\sigma_0$ is locally stable if and only if $M(R(\mathbf{r}_0))$ is unique and has center $\mathbf{r}_0$ and radius $\sigma_0$.
\label{theor:mainRV}
\end{theorem}

\begin{proof}

First, we must prove that in a hard sphere system, a particle with center $\mathbf{r}_0$ and radius $\sigma_0$ being locally stable implies that $M(R(\mathbf{r}_0))$ is unique and has center $\mathbf{r}_0$ and radius $\sigma_0$. Following the logic of the proof of Theorem~\ref{theor:noUniqueNoStable}, we assume that $M(R(\mathbf{r}_0))$ has center $\mathbf{r}_1$ and radius $\sigma_1$ with $\mathbf{r}_1 \neq \mathbf{r}_0$ and $\sigma_1 \neq \sigma_0$ and find a contradiction to show that $\mathbf{r}_1 \neq \mathbf{r}_0$  and $\sigma_1 = \sigma_0$. If $\sigma_1 < \sigma_0$, then this does not correspond to the maximum inscribed sphere. If $\sigma_1 \ge \sigma_0$, then $\overline{B}_{\sigma_0}(\mathbf{r}_1) \subset \overline{B}_{\sigma_1}(\mathbf{r}_1) \subset R(\mathbf{r}_0)$ and thus by Theorem~\ref{theor:convexBall} $\overline{B}_{\sigma_0}(\mathbf{r}_2) \subset R(\mathbf{r}_0)$ for all $\mathbf{r}_2$ on a straight line between $\mathbf{r}_0$ and $\mathbf{r}_1$. But because $\overline{B}_{\sigma_0}(\mathbf{r_0})$ is locally stable, no translations or dilations exist which remain in $R(\mathbf{r}_0)$, so $\mathbf{r}_1 = \mathbf{r}_0$ and $\sigma_1 = \sigma_0$.

Second, we must prove that in a hard sphere system, for a particle centered at $\mathbf{r}_0$ with radius $\sigma_0$,  $M(R(\mathbf{r}_0))$ being unique and having center $\mathbf{r}_0$ and $\sigma_0$ implies that the particle is stable. This follows immediately from Corollary~\ref{cor:uniqueMIS}, as the particle has $d+1$ non-cohemispheric points of contact with $R(\mathbf{r}_0)$, which by Corollary~\ref{cor:RVTangentPlane}, correspond to contacts with neighboring particles. Thus, by Corollary~\ref{cor:nonCohemisphericForces}, the particle centered at $\mathbf{r}_0$ with radius $\sigma_0$ is stable.
\end{proof}

\begin{figure}[htp!]
\includegraphics[width=0.9\linewidth]{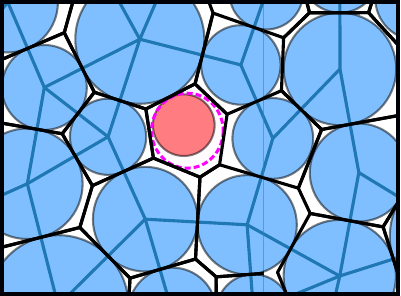}
\caption{The radical Voronoi diagram (black lines) is computed for a set of particles of different radii with contacts displayed as blue lines. All blue particles, labelled $i$, have $M(R(\mathbf{r}_i)) = \overline{B}_{\sigma_i}(\mathbf{r}_i)$ and are thus stable. The red particle, which we will call $0$ is a rattler, and its MIS is shown as a dashed magenta line. We see clearly that $M(R(\mathbf{r}_0)) \neq \overline{B}_{\sigma_0}(\mathbf{r}_0)$}
\label{fig:RVstability}
\end{figure}

\section{Algorithmic complexity}
\label{sec:complexity}

Theorems~\ref{theor:mainCH} and~\ref{theor:mainRV} provide a natural recursive algorithm for determining the stable set of particles in a packing, and, through its complement, the set of rattlers. The algorithm begins with a tentative statement that all particles are stable, and it loops over each particle testing for stability, taking the function $\mathrm{isStable}(i)$ from either Theorem~\ref{theor:mainCH}, Theorem~\ref{theor:mainRV}, or Eq. 12 of Ref.~\cite{donev_linear_2004}, considering only the stable set of particles. The algorithm ends when no changes are made to the stable list in a full loop. The structure of the algorithm is similar to that of Ref.~\cite{donev_linear_2004}, and as expected, it produces an identical stable list.

\begin{algorithm}[H]
\caption{Global Stability Algorithm}
\begin{algorithmic}[1]
\State $i \in \mathrm{stableList} \,\,\forall i$
\State $\mathrm{unstableList} = \emptyset$
\State $\mathrm{flip} \gets true$
\While{flip} 
    \State $\mathrm{flip} \gets false$
    \For{$i \in \mathrm{stableList}$}
       \If{isUnstable(i)}
            \State $\mathrm{flip} \gets true$
            \State Move $i$ from stableList to unstableList
       \EndIf
    \EndFor
\EndWhile
\State{Return stableList}
\end{algorithmic}
\end{algorithm}

The worst-case scenario for this algorithm is a packing in which only a single particle is initially unstable, but its removal destabilizes one of its neighbors, and so on.  Such a situation will require $N$ iterations through the algorithm, each of which takes $\mathcal{O}(N)$ time, yielding a total worst case runtime of $\mathcal{O}(N^2)$. We note, however, that no typical case approaches this complexity. The method of Ref.~\cite{lerner_simulations_2013}, meanwhile, scales as at least $\mathcal{O}(d^3N^3)$~\cite{charbonneau_jamming_2023}.

The only difference between the methods of Ref.~\cite{donev_linear_2004}, Theorem~\ref{theor:mainCH}, and Theorem~\ref{theor:mainRV} is the speed of the function $\mathrm{isStable}(i)$.
For a particle with $n$ contacting particles (where $n \sim \mathcal{O}(d)$), the linear programming method scales as $\mathcal{O}(n^{2+a})$ where $a=\frac{1}{18}$~\cite{jiang_faster_2021} while the convex hull scales as $\mathcal{O}(n^{\lfloor d/2 \rfloor})$ in the worst case scenario, where $\lfloor \cdot \rfloor$ is the floor function~\cite{chazelle_optimal_1993}. The radical Voronoi diagram for an individual cell can be computed in $\mathcal{O}(n^{\lceil d/2 \rceil})$ where $\lceil \cdot \rceil$ is the ceiling function. Thus while the radical Voronoi method is slower than the convex hull method in odd dimensions, it is of the same order in even dimensions. The calculation of the MIS is then either a linear programming problem~\cite{gritzmann_inner_1992, gritzmann_computational_1993} or a minimization problem~\cite{morse_geometric_2014} whose complexity has not yet been interrogated. The worst case scenario then makes this calculation the rate determining step, and it is thus no faster than the linear programming methods of Ref.~\cite{donev_linear_2004}. By comparison, we see that the convex hull algorithm is faster than the linear programming algorithm for at least $d<6$.

\section{Further extensions}
\label{sec:conclusion}

We have shown that the convex hull and the radical Voronoi cell can be used to quickly determine the stability of individual spheres in a packing with only minimal requirements on the interparticle potential. It is straightforward to show that the construction can be applied more generally in a variety of cases. Here, we list several:

\begin{enumerate}

\item In a spring network under compression, an individual node is unstable if it is on the surface of the convex hull of its connecting nodes.

\item A particle of any shape is unstable if the only forces acting on it are point forces directed towards its center of mass, and the center of mass is on the surface of the convex hull of the contact points and the center of mass.

\item In Mari-Kurchan (MK) interactions~\cite{mari_jamming_2009, mari_dynamical_2011}, where the distance between particles is given by ${h_{ij}^{\mathrm{MK}} = \frac{\abs{\mathbf{r}_i - \mathbf{r}_j + \mathbf{\Lambda}_{ij}}}{\sigma_i + \sigma_j}}$ where $\mathbf{\Lambda}_{ij}$ is a random vector with $\mathbf{\Lambda}_{ij} = -\mathbf{\Lambda}_{ji}$, a particle $\mathbf{r}_0$ with contacts $\{\mathbf{r}_j\}$ is unstable if $\mathbf{r}_0 \in \partial \mathrm{Conv}(\mathbf{r}_0,\{\mathbf{r}_j + \mathbf{\Lambda}_{ij}\})$. This method was used in Ref.~\cite{charbonneau_jamming_2023}. Note that this is true despite not technically being a central force potential.

\item Several recent studies have analyzed soft sphere systems during energy minimization~\cite{charbonneau_jamming_2023, stanifer_avalanche_2022, nishikawa_relaxation_2021, nishikawa_relaxation_2022, manacorda_gradient_2022}, wherein it may be important to study the evolution of rattlers and stable subsystems. Here, the convex hull theorem may be used, with the additional caveat that a particle is only locally stable if the sum of all forces acting on it is zero, and if the forces acting on it span $\mathbb{R}^d$.

\item Following the logic of Sec.~\ref{sec:rvProof}, we conjecture that Theorem~\ref{theor:mainRV} also holds for additively-weighted Voronoi cells and any generalization of Voronoi cells $G$ for which the contact point of two hard spheres ($i$ and $j$) is contained on the surface of the generalized Voronoi cell, i.e. $\overline{B}_{\sigma_i}(\mathbf{r}_i) \cap \overline{B}_{\sigma_j}(\mathbf{r}_j) \in \partial G(\mathbf{r}_i)$ and $\overline{B}_{\sigma_i}(\mathbf{r}_i) \cap \overline{B}_{\sigma_j}(\mathbf{r}_j) \in \partial G(\mathbf{r}_j)$. However, these cells are generically non-convex, and so some of the tools we have used do not suffice. 

\end{enumerate}

These extensions show the utility of our methods, which extend beyond simple sphere packings. It is our hope that this work not only provides a simple computational tool, but helps to illuminate the interplay between geometry and mechanical rigidity.

\bibliographystyle{apsrev4-2}
\bibliography{convexHull,footnotes}

\end{document}